\documentclass{article}
\usepackage[affil-it]{authblk}
\usepackage{graphicx}
\usepackage[space]{grffile}
\usepackage{latexsym}
\usepackage{amsfonts,amsmath,amssymb}
\usepackage{url}
\usepackage[utf8]{inputenc}
\usepackage{hyperref}
\hypersetup{colorlinks=false,pdfborder={0 0 0}}
\usepackage{textcomp}
\usepackage{longtable}
\usepackage{multirow,booktabs}

\usepackage{geometry}
\usepackage{amsthm}

\newtheorem{proposition}{Proposition}

\begin{document}

\title{SAT as a game}

\author{Olivier Bailleux}
\affil{Université de Bourgogne Franche-Comté}

\date{}

\maketitle 

\begin{abstract}
We propose a funny representation of SAT. While the primary interest is to present propositional satisfiability in a playful way for pedagogical purposes, it could also inspire new search heuristics.
\end{abstract}

\section{The colored tokens game}
The game board consists of boxes containing round and square tokens of different colors. For each color, the player must remove either all the square tokens or all the round ones. The game is won if there is at least one token in each box. The result is the same regardless of the order in which the tokens are removed. It depends only on the shape removed for each color.

\begin{figure}[h!]
\begin{center}
\includegraphics[width=0.7\columnwidth]{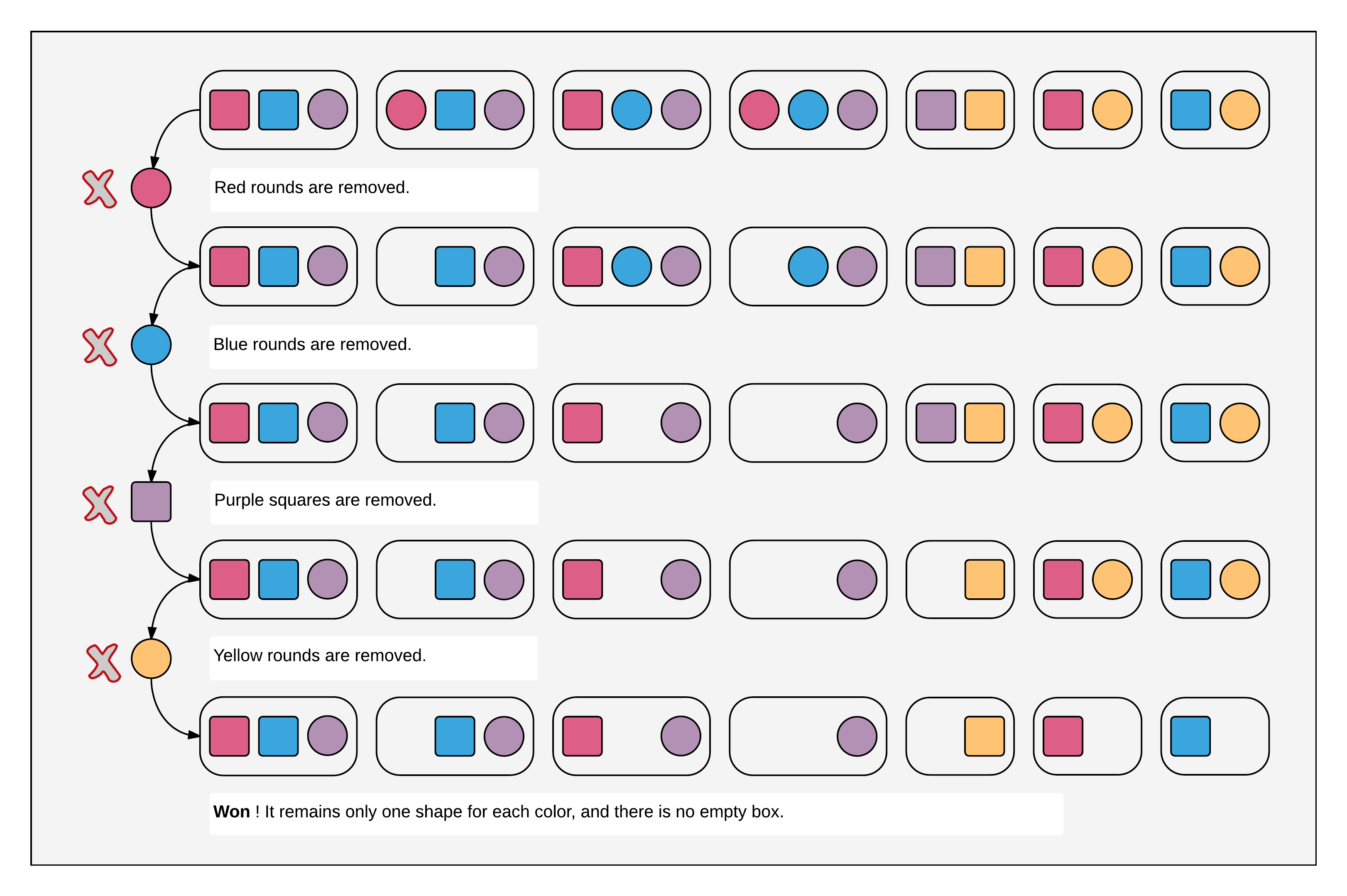}
\caption{An example of game play.}
\label{aparty}
\end{center}
\end{figure}

The initial configuration (boxes and tokens) is said to be an \emph{instance} of the game. It is said to be \emph{feasible} if and only if it admits at least one solution.

\section{SAT representation}

\subsection{SAT}

In the context of propositional logic, a \emph{literal} is either a propositional variable $v$ or its negation $\neg v$, a \emph{clause} is a disjunction of literals $l_1 \vee \cdots \vee l_m$, and a CNF formula is a conjunction of clauses $q_1 \wedge \cdots \wedge q_n$. An \emph{interpretation} of a formula $\Sigma$ is an assignment of a truth value (true or false) to each variable occurring in $\Sigma$. An interpretation is said to \emph{satisfy} a formula $\Sigma$ if and only if it evaluates this formula to true according to the standard semantic of the operators $\neg, \vee, \wedge$. A formula is said to be \emph{satisfiable} if and only if there exists at least one interpretation that satisfies it.

SAT is the problem of determining whether a given CNF formula is satisfiable or not. For more information, see, for example \cite{biere2009handbook}. 

\subsection{Representing a SAT instance with the colored tokens game}

SAT reduces to the colored tokens game as follows. Each box corresponds to a clause and each color corresponds to a variable. Any square token represents a positive literal, and any round token represents a negative literal.

For example, the following formula encodes the initial configuration of the figure \ref{aparty}. The variable $r$ stands for the color red, $b$ for blue, $p$ for purple, and $y$ for yellow, respectively.

\[ 
(r \vee b \vee \neg p) \wedge 
(\neg r \vee b \vee \neg p) \wedge
(r \vee \neg b \vee \neg p) \wedge
(\neg r \vee \neg b \vee \neg p) \wedge
(p \vee y) \wedge
(r \vee \neg y) \wedge
(b \vee \neg y)
\]

Clearly, any interpretation satisfying the resulting formula (if applicable) corresponds to a solution of the underlying game instance.

\section{A variant  \label{section-variant}}

Here is a variant of the game based on alternative rules : the player removes one token at a time and can either remove a square or a round token of any color. The game ends when each box contains exactly one token. The player wins if and only if all the tokens of each color have the same shape.

SAT reduces to this variant in the same way as to the original game thanks to the following property.

\begin{proposition}
Any instance of the variant is feasible if and only if the corresponding instance of the original game is feasible.
\end{proposition}

\begin{proof}[Sketch of the proof]
~
\begin{enumerate}
\item Any solution of the variant can be obtained from a solution of the original problem by removing all tokens but one from each box.
\item Given a solution $s_v$ of the variant, a solution of the original game can be obtained in the following way: from the initial state, for each color, the tokens of the shape occurring in $s_v$ are kept, and the other ones are removed.
\end{enumerate}
\end{proof}

\section{Interests of the game and related works}

There are several examples of links between SAT and games in the literature. For example, SAT encodings for the popular Sudoku game are presented in \cite{DBLP:conf/isaim/LynceO06} and \cite{kwon2006optimized}. Conversely, the famous Candy Crush game was proved to be NP-hard thanks to a reduction from SAT in \cite{DBLP:journals/corr/Walsh14}. There is also a WEB application, "The SAT game" from Olivier Roussel (CRIL université d'Artois), which proposes to solve a SAT instance represented as an integer matrix where each line is a clause. Variables and literals are represented by an integers in the same way as in the largely used DIMACS format. 

The aim of the colored tokens representation is to be easier to grasp for humans. It could be used for the vulgarization of propositional satisfiability, but also to analyze the strategies of the best players and to use them in designing new search heuristics or branching rules. The colored tokens representation can also inspire new search heuristics by itself. For example, a local search solver inspired from the variant presented section \ref{section-variant} could explore a search space where each configuration selects only one literal per clause, with a cost function which penalizes the variables occurring both in positive and negative forms.

\bibliographystyle{plain}
\bibliography{bibli}

\end{document}